\documentclass[aps,pra,twocolumn,floatfix,nofootinbib,superscriptaddress,longbibliography,10pt]{revtex4-2} 
\usepackage{bm,graphicx,mathrsfs,amsmath,amssymb,mathtools,makecell,bbm, amsthm,dsfont,color,times,txfonts,nicefrac,framed,enumitem,tikz,physics,wrapfig,amsfonts,lipsum,nicematrix}
\usepackage[most]{tcolorbox}
\usepackage[dvipsnames]{xcolor}
\usepackage[colorlinks=true,linkcolor=NavyBlue,citecolor=NavyBlue]{hyperref}

\usepackage{etoolbox}

\makeatletter
\patchcmd{\@bibdataout@aps}
  {author="08"}
  {author="48"}
  {}{}
\makeatother

\hypersetup{urlcolor=NavyBlue}

\makeatletter
\def\blfootnote{\gdef\@thefnmark{}\@footnotetext}
\makeatother

\renewcommand{\v}[1]{\ensuremath{\boldsymbol #1}}
\newcommand{\ms}[1]{\textsf{#1}}
\newcommand{\iden}{\mathbbm{1}}
\newcommand{\STAB}{\operatorname{STAB}}
\newcommand{\conv}{\operatorname{conv}}
\newcommand{\bx}{\bar{x}}
\newtheorem{thm}{Theorem}

\newtheorem{lem}[thm]{Lemma}

\newtheorem{cor}[thm]{Corollary}

\begin{document}

\title{When Symmetry Suppresses Magic}

\author{A. de Oliveira Junior}
	\affiliation{Center for Macroscopic Quantum States bigQ, Department of Physics,
Technical University of Denmark, Fysikvej 307, 2800 Kgs. Lyngby, Denmark}
\email{alexssandredeoliveira@gmail.com}
\author{Jake Xuereb}
	\affiliation{Vienna Center for Quantum Science and Technology, Atominstitut, TU Wien, 1020 Vienna, Austria}
\author{Rafael A. Macêdo}
    \affiliation{Department of Mathematical Sciences, University of Copenhagen, Universitetsparken 5, 2100 Copenhagen,
Denmark}
\author{Jonatan Bohr Brask}
    \affiliation{Center for Macroscopic Quantum States bigQ, Department of Physics,
Technical University of Denmark, Fysikvej 307, 2800 Kgs. Lyngby, Denmark}
\author{Rafael Chaves}
    \affiliation{International Institute of Physics, Federal University of Rio Grande do Norte, 59078-970, Natal, RN, Brazil}
\date{\today}

\begin{abstract}
Nonstabilizerness is a critical resource for quantum advantage, but evaluating it, especially for mixed states, requires superexponentially many samples in system size, making the problem NP-hard. While symmetries are known to reduce this complexity, it is unclear whether they also restrict the \textit{amount} of magic. In this work, we provide an explicit example of such a symmetry by proving that the Robustness of Magic (RoM) for $N$-qubit \ms{X}-states is at most $\sqrt{3}$. Leveraging this symmetry constraint, we introduce a computationally efficient method to lower-bound the RoM of arbitrary many-body states, demonstrating its utility on the ground states of a spin-$\tfrac12$ Hamiltonian with system sizes well beyond the reach of exact evaluation. Furthermore, for Hamiltonians whose equilibrium states are \ms{X}-states, we analytically derive the critical temperature at which magic emerges. Our work shows why certain symmetries constrain magic while others do not, provides a scalable lower-bounding method, and identifies a nontrivial regime where many-body nonstabilizerness is analytically solvable.
\end{abstract}

\maketitle

\paragraph*{Introduction---}

Symmetries in nature provide regularities that allow us to make sense of natural phenomena. They give rise to conservation laws~\cite{Noether01011971}, simplify calculations~\cite{Eckart1930,Wigner1959}, and reduce the complexity of simulating many-body systems~\cite{vidal_07,vidal_08,Singh2010,Renner2007}. This naturally leads us to ask how symmetries affect quantum resources. Of particular interest is nonstabilizerness (or `magic'), the resource that captures the hardness of classically simulating quantum states~\cite{gottesman1998heisenberg,aaronson2004improved,bravyi2005universal,bravyi_16,lima_24,leone2026,qian_24}.

While magic typically grows extensively with the system size $N$~\cite{Liu2022} and requires $4^N$ Pauli expectation values to be quantified~\cite{veitch2014resource,Howard2017,leone2022stabilizer}, symmetries can mitigate this exponential experimental and computational cost. For instance, permutation symmetry reduces the number of required observables to $\mathcal{O}(N^3)$~\cite{Passarelli2024}, while gauge constraints, conserved charges, and symmetries of the stabilizer polytope offer similar numerical simplifications~\cite{Cepollaro2024,Ruyter2025b,iannotti2026,sabharwal2026stuffmagiconcompact,Heinrich2019robustnessofmagic,Hamaguchi2024handbookquantifying}. Beyond simplifying computations, symmetries can either protect or suppress nonstabilizerness~\cite{Ellison2021symmetryprotected,iannotti2026}. However, the use of symmetries to quantify magic has so far remained largely numerical, leaving open when symmetries admit analytical solutions and impose universal bounds on the magic a quantum state can possess. Here, we show that parity symmetry accomplishes both by making nonstabilizerness exactly computable while placing fundamental limits on the amount of magic a many-body state can carry.

Consider a system of $N$ interacting spin-$\tfrac12$ particles whose state obeys a parity symmetry. Neighboring spins may point in the same or opposite directions along the $z$-axis, but this symmetry allows quantum coherence only between configurations that differ by flipping all spins at once. Such symmetry arises naturally in both equilibrium and nonequilibrium settings. Examples include Gibbs states of diagonal Ising Hamiltonians and two-site marginals of equilibrium states of a broader class of spin-chain models, such as anisotropic $XY$ and transverse-field Ising chains~\cite{Osborne2002,Sarkar2020}. It also appears during the dynamics of quantum impurity models~\cite{Cavalcante2025}, in nonequilibrium steady states of quantum thermal machines~\cite{brask2015autonomous,Khandelwal2020} and in protocols converting athermality into entanglement~\cite{deOliveiraJunior2024}. The density matrix of such a system takes the form of an $\ms{X}$-state~\cite{YuEberly2007,Rau2009,Vinjanampathy2010,huber_12} in the $z$-basis, with nonzero entries only on the diagonal and anti-diagonal. Our main result is that this parity symmetry strongly constrains nonstabilizerness, giving a closed-form expression and a universal upper bound for the robustness of magic (RoM)~\cite{Howard2017}. 

\begin{figure*}
    \centering
    \includegraphics{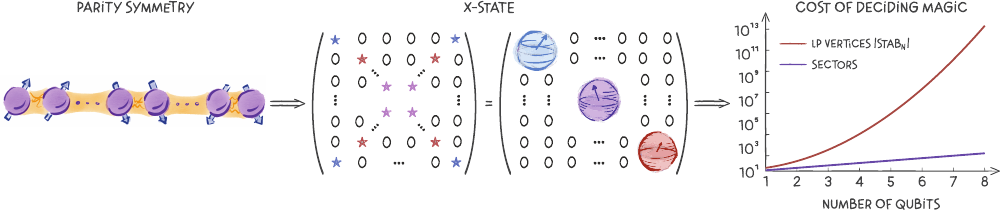}
    \caption{\emph{This paper in a nutshell.} A state obeying the parity symmetry has nonzero matrix elements only on the diagonal and antidiagonal. Such a state decomposes into two-dimensional sectors, each behaving as an effective qubit. Magic can consequently be identified and quantified separately in each sector. This replaces the linear programme over the super-exponentially pure stabilizer states with fixed-size calculation per sector. The number of sectors remains exponential in the number of qubits, but the resulting calculation is linear in the description of the state. }
    \label{F:nuthsell-paper}
\end{figure*}

We first show that $N$-qubit $\ms X$-states decompose into independent blocks, reducing the many-body problem to a set of independent single-qubit problems. The resulting formula requires a constant number of operations per block, so its cost scales linearly with the state's dimension while avoiding optimization over the superexponentially many pure stabilizer states (see Fig.~\ref{F:nuthsell-paper} for a summary). We then prove that the RoM of any $\ms X$-state is bounded above by $\sqrt{3}$, independently of the system size, showing that this symmetry permits at most a single qubit's worth of magic. Next, we introduce the $\ms{X}$-twirl to bound the RoM of arbitrary states from below. We apply it to ground states of the transverse-field Ising model with up to 15 spins, far beyond the reach of direct computation. Turning to dynamics, we classify the Hamiltonians compatible with this symmetry and use our framework to derive an analytical expression for the critical temperature at which magic emerges in equilibrium.

\paragraph*{Parity symmetry \& antipodal sectors---} The symmetry we exploit records the relative orientation of neighboring spins along the $z$ axis via the commuting Pauli operators $G_i:=Z_i Z_{i+1}$ for $i=1, \ldots, N-1$. Each $G_i$ distinguishes whether the spins at sites $i$ and $i+1$ are aligned or anti-aligned. Together, the $N-1$ eigenvalues determine a bit string up to a global spin flip, so each joint eigenspace is spanned by a pair of bitwise-complementary states.

For $x = (x_1, \ldots, x_N)$ let $\bx = (1\oplus x_1, \ldots, 1\oplus x_N)$ denote its bitwise complement. The strings $x$ and $\bx$ have the same relative-parity eigenvalues, since flipping every spin leaves each $G_i$ unchanged. We call $(x,\bx)$ an antipodal pair, since the two strings differ at every site. The sector $\mathcal{B}_x:=\operatorname{span}\{\ket{x},\ket{\bx}\}$ is the eigenspace in which each $G_i$ has eigenvalue $\varepsilon_i^x:=(-1)^{x_i\oplus x_{i+1}}$. Its projector is $\Pi_x:=\prod_{i=1}^{N-1}\tfrac12(\iden+\varepsilon_i^x G_i)$. To label each sector once, we choose the representative with $x_1=0$ and introduce $\mathcal{X}_N:=\{x\in\{0,1\}^N:x_1=0\}$. Since distinct sectors have different relative-parity eigenvalues, a state commuting with every $G_i$ cannot contain coherence between them, and therefore
\begin{equation}\label{Eq:X-states}
    [\rho,G_i]=0 \:\:\text{for all $i$}\:\Longleftrightarrow\: \rho =\bigoplus_{x\in\mathcal{X}_N} \begin{pmatrix}
        r_x & \alpha_x \\ \alpha_x^* & r_{\bx}
    \end{pmatrix}_{\mathcal{B}_x},
\end{equation}
with $r_x:=\langle x|\rho|x\rangle$ and $\alpha_x:=\langle x|\rho|\bx\rangle$. Normalization requires $\sum_{x\in\mathcal{X}_N}(r_x+r_{\bx})=1$, while positivity requires $r_x, r_{\bx} \geq 0$ and $|\alpha_x|^2 \leq r_x r_{\bx}$ in every sector. This is the structure of an $\ms{X}$-state~\cite{YuEberly2007, Rau2009}. In the computational basis, only diagonal elements and coherences between bitwise-complementary states can be nonzero. Equivalently, the $G_i$ are the stabilizer generators of the $N$-qubit repetition code \cite{gottesman1997stabilizer}, and the $\mathcal{B}_x$ are its $2^{N-1}$ syndrome subspaces, each carrying one effective qubit. Within each sector, we set $\ket{0_L}=\ket{x}$, $\ket{1_L} = \ket{\bx}$, and define the effective Pauli operators
\begin{equation}\label{Eq:logical-paulis}
    \begin{split}
    X^{(x)}&=\ket{x}\!\bra{\bx}+\ket{\bx}\!\bra{x},\\ 
    Y^{(x)}&=-i(\ket{x}\!\bra{\bx}-\ket{\bx}\!\bra{x}),\\
    Z^{(x)}&=\ketbra{x}-\ketbra{\bx}.
\end{split}
\end{equation}
Writing $\v \sigma^{(x)} = (X^{(x)},Y^{(x)},Z^{(x)})$, the restriction of $\rho$ to this sector takes the Bloch form $\rho_x:=\Pi_x \rho \Pi_x = \tfrac12(p_x \Pi_x + \v a_x \cdot \v\sigma^{(x)})$, where $p_x :=r_x+r_{\bx}$ is the sector population and $\v a_x:= (2\operatorname{Re} \alpha_x, -2\operatorname{Im} \alpha_x, r_x-r_{\bx})$ is the subnormalized Bloch vector of the effective qubit. For later use, we define $L_x:=\|\v a_x\|_1$, $m_x:=\min\{r_x, r_{\bx}\}$ and $s_x:=|\operatorname{Re}\alpha_x| + |\operatorname{Im}\alpha_x|$, so that $L_x=2s_x+|r_x-r_{\bx}|$ and $p_x - |r_x-r_{\bx}| = 2m_x$. 

\paragraph*{Magic under parity symmetry--} To quantify the magic, we use the robustness of magic (RoM)~\cite{Howard2017}, which applies to both pure and mixed states. Operationally, the RoM dictates the sampling overhead in classical simulations based on signed stabilizer decompositions~\cite{Pashayan2015,Howard2017}. Let $\operatorname{STAB}_N$ denote the set of $N$-qubit pure stabilizer projectors. In its stabilizer norm form, the RoM is \mbox{$\mathcal{R}(\rho) = \min \qty{\sum_k |l_k| : \rho = \sum_k l_k \tau_k, \tau_k \in \operatorname{STAB}_N}$} with real coefficients $l_k$. It satisfies $\mathcal{R}(\rho)\geq 1$ for every state, with $\mathcal{R}(\rho)>1$ when $\rho$ is magical. Evaluating the RoM directly requires optimization over the set of stabilizer projectors, whose cardinality grows superexponentially with $N$. 

For $\ms{X}$-states, parity symmetry removes the need for this optimization. Our main result, proven in Appendix~\ref{S:SM-2}, gives a closed-form expression for the RoM. Namely, for any $N$-qubit $\ms{X}$-state
\begin{align}\label{Eq:main-result-RoM}
   \mathcal{R}(\rho) &= \sum_{x\in \mathcal{X}_N} \max\{p_x, L_x\} = 1+2\sum_{x\in \mathcal{X}_N}(s_x-m_x)_+,
\end{align}
where $(\bullet)_+ := \max\{0,\bullet\}$. It also gives the exact membership criterion (see~Appendix~\ref{S:SM-2} for details)
\begin{equation}\label{Eq:necessary-and-sufficient condition}
    \rho \in \conv(\STAB_N) \Longleftrightarrow s_x \leq m_x \:\:\text{for all $x$}.
\end{equation}
An $\ms{X}$-state is therefore magical when the normalized effective qubit state in at least one populated sector lies outside its stabilizer octahedron. The criterion follows because the six effective Pauli eigenstates in each sector are physical $N$-qubit stabilizer states (see Lemma~\ref{Lem:stabilizer-states-antipodal} in App.~\ref{S:SM-1}), whereas projecting any global convex stabilizer decomposition onto a fixed sector constrains its subnormalized effective Bloch vector to the corresponding octahedron (see Lemma~\ref{Lem:sector-stabilizer-projection} in App.~\ref{S:SM-1}). The RoM formula also has a direct interpretation. For each populated sector, let $\hat{\rho}_x:=\rho_x/p_x$ denote the corresponding normalized effective-qubit state. Eq.~\eqref{Eq:main-result-RoM} then states that the RoM of the $N$-qubit state $\rho$ is the average of the sector RoMs, $\mathcal{R}(\rho) =\sum_{x:p_x>0} p_x \mathcal{R}(\hat\rho_x)$, weighted by the sector probabilities. The sectors contribute independently, so magic does not accumulate across them. As a consistency check, for a single nonzero GHZ coherence, the membership criterion reduces to that obtained in Ref.~\cite{Cao2026}, while its exact RoM formula for real coherence is recovered as a special case.

Reducing the state to effective qubits also limits how much magic it can have. Since the RoM of every normalized effective-qubit state is at most $\sqrt{3}$ and $\sum_xp_x=1$, the preceding identity implies that $\mathcal{R}(\rho) \leq \sqrt{3}$ for every $N$ (see Corollary~\ref{cor:uniform-bound}). As shown in Appendix~\ref{S:SM-2}, equality holds if and only if every populated sector is pure after normalization and has effective Bloch vector $\tfrac{\v{a}_x}{p_x} = \pm \tfrac{1}{\sqrt{3}}(1,1,1)$, where each component sign may be chosen independently in every sector. Maximal magic can be attained either by a pure state supported in a single sector or, for $N\geq2$, by a mixed state distributed over several orthogonal sectors. The bound does not rule out many-body coherence. For example, a sector may contain a GHZ-like superposition of two macroscopically distinct strings. What is suppressed is the accumulation of magic across sectors. The symmetry confines any surviving coherence within each sector to one effective qubit, and the different sectors enter the computation only via the average above.

\paragraph*{A magic witness for arbitrary states--} The same symmetry gives a certified bound for an arbitrary state $\sigma$. Let
\begin{equation}
    \Lambda_{\ms{X}}(\sigma) = \frac{1}{2^{N-1}}\sum_{\substack{w\in\{0,1\}^N \\ |w| \:\text{even}}}Z_w \sigma Z_w,
\end{equation}
where $Z_w:=\bigotimes_{i=1}^N Z_i^{w_i}$ and $|w|:=\sum_{i=1}^N w_i$ is the Hamming weight of $w$. The even-weight $Z_w$ form a group generated by $G_i$, so $\Lambda_{\ms{X}}$ is the twirl over the symmetry. Averaging over them cancels all matrix elements except those on the diagonal and antidiagonal, which are left unchanged, so $\Lambda_{\ms{X}}(\sigma)$ is an $\ms{X}$-state. Since the twirl is a convex mixture of Clifford conjugations, it cannot increase the RoM. Therefore, for any Clifford unitary $C$,
\begin{equation}\label{Eq:X-witness}
    \mathcal{R}(\sigma) \geq \mathcal{R}[\Lambda_{\ms{X}}(C\sigma C^\dagger)],
\end{equation}
and the right-hand side, given in closed form by Eq.~\eqref{Eq:main-result-RoM}, is a certified lower bound for every $C$. Whenever it exceeds one, $\sigma$ is magical. Importantly, the bound should be read as an efficient certificate and not as a quantifier of the total amount of magic. The choice of Clifford operation can nevertheless strengthen the bound. Eq.~\eqref{Eq:X-witness} holds for every $C$, so optimizing over Clifford operations changes only the value of the lower bound. Since $\Lambda_{\ms{X}}$ keeps only the diagonal and antidiagonal elements in the computational basis, different Clifford operations select different components of $\sigma$. We call the bound obtained using a specified $C$ the twirl bound, and the bound obtained after optimization the optimal-twirl bound. 

As an illustration, we consider the ground state of the transverse-field Ising model with periodic boundary conditions \mbox{$H(g)=\sum_{i=1}^N Z_i Z_{i+1} - g\sum_{i=1}^NX_i$}, with $Z_{N+1} \equiv Z_1$. For $N=4$, Figure~\ref{F:X-twirl-Ising} compares the RoM computed directly, the reduced-RoM of Ref.~\cite{varela2026predicting}, constructed from the eight Pauli observables appearing in the Hamiltonian, and our optimized twirl. The solid curves show these three results. For $N=4$, the optimization can be performed exhaustively, so the optimal-twirl bound curve is the global optimum over Clifford operations at every value of $g$. Beyond $N=5$, both direct calculations of the RoM and exhaustive optimizations over Clifford operations become computationally impractical. To show the usefulness of our bound in certifying the presence of magic, the dashed curves show certified twirl bounds for $N=10$ and $15$, using one fixed Clifford operation for each system size, selected via a randomized search (see \hyperref[Sec:Endmatter]{End Matter} for details). Finally, the peak near $g=1$ (dashed vertical orange line) is a finite-size signature of the quantum phase transition~\cite{PFEUTY197079}, arising from competition between the Ising interaction and the transverse field. The twirl bounds reproduce this peak and certify nonstabilizerness near the critical point even when the full RoM cannot be calculated. This is consistent with previous results for RoM in reduced states and stabilizer Rényi entropy in many-body ground states near Ising criticality~\cite{Sarkar2020,Oliviero2022}.
\begin{figure}[t]
    \centering
    \includegraphics{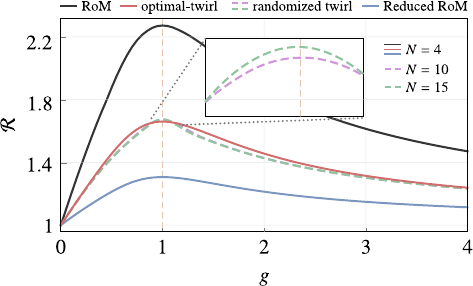}
    \caption{\emph{Witnessing magic in the transverse-field Ising model.} Ground state RoM for the periodic transverse-field Ising model as a function of $g$. Shown are the RoM (black), the optimal-twirl bound (red), and the reduced-RoM lower bound of Ref.~\cite{varela2026predicting} (blue). Both bounds certify magic for $g>0$, with the $\ms{X}$-twirl being the stronger certificate. For $N=10$ and $N=15$, the dashed curves use a fixed Clifford operation selected by randomized search for each size. Note that the ground state here is not an $\ms{X}$-state, and the twirl is what brings it into the solvable family.}
    \label{F:X-twirl-Ising}
\end{figure}

\paragraph*{Which Hamiltonians are nearest-neighbor parity symmetric?--} So far we have imposed the symmetry directly on the state. We now ask which many-body Hamiltonians produce it. The Gibbs state $\rho_\beta \propto e^{-\beta H}$ is an $\ms{X}$-state at all inverse temperatures $\beta\geq0$ if and only if $[H,G_i]=0$ for all $i$. To determine which interactions satisfy this condition, we expand $H$ in the Pauli-string basis. Since the Pauli strings are linearly independent, $H$ commutes with every $G_i$ if and only if every Pauli string appearing in the expansion does so. Consider a string $P=P_1\otimes \cdots \otimes P_N$ and define $t_i=0$ when $P_i \in \{\iden, Z\}$ and $t_i = 1$ when $P_i\in \{X,Y\}$. Because $X$ and $Y$ anticommute with $Z$, whereas $\iden$ and $Z$ commute with it, we have $G_i P = (-1)^{t_i+t_{i+1}}P G_i$. Thus, $P$ commutes with every $G_i$ when all $t_i$ coincide (see Appendix~\ref{S:SM-4} for details). Therefore, every allowed string is either diagonal, containing only $\iden$ and $Z$, or fully transverse, containing an $X$ or a $Y$ on every site. It follows that if $H$ is a $k$-local Hamiltonian with $k<N$ and $[H,G_i]=0$ for all $i$, then $H$ is diagonal in the computational basis. Consequently, its Gibbs state is a convex mixture of computational basis-states with $\mathcal{R}(\rho_\beta)=1$ at all temperatures. 

Magic in a nearest-neighbor parity-symmetric Gibbs state can arise only if $H$ contains a global transverse coupling, defined as an $N$-body Pauli string whose factor on every site is either $X$ or $Y$. With such a coupling, checking for magic becomes a single-qubit problem in each sector. In $\mathcal{B}_x$, the Hamiltonian takes the form $H_x=\epsilon_x\Pi_x - \v h_x\cdot \v \sigma^{(x)}$. As derived in Appendix~\ref{S:SM-4}, for $h_x>0$, the Bloch vector of the normalized sector Gibbs state points along $\hat{\v{h}}_{x}$ and has length $\tanh(\beta h_x)$, where $h_x=\|\v h_x\|_2$. For $h_x=0$, the normalized sector state is maximally mixed. Cooling therefore drives the Bloch vector outward in a fixed direction. The sector becomes magical when this vector crosses the boundary of the stabilizer octahedron, giving the critical inverse temperature
\begin{equation}
\beta_{\text{crt}}^{(x)} = 
\begin{cases} 
\dfrac{1}{h_x}\operatorname{atanh}\qty(\frac{h_x}{\|\v{h}_x\|_1}), &  \|\v h_x\|_1> h_x, \\[4pt]
+\infty, &  \|\v h_x\|_1= h_x.
\end{cases}
\end{equation}
The sector becomes magical once $\beta$ exceeds $\beta_{\text{crt}}^{(x)}$, which is finite if and only if $\|\v{h}_x\|_1 > \|\v{h}_x\|_2$, or equivalently, if at least two components of the logical field are nonzero. A single nonzero logical-field component drives the encoded qubit along an axis of the octahedron and can never take it outside, no matter how far the system is cooled. Magic requires competing logical-field components. For $h_x>0$, the ratio  $\|\v{h}_x\|_1/\|\v{h}_x\|_2\in [1,\sqrt{3}]$ determines the dimensionless threshold $\beta_{\text{crt}}^{(x)} h_x$ and the zero-temperature RoM within that sector, while $h_x$ sets the temperature scale. At every finite inverse temperature, all sectors have positive population. The Gibbs state becomes magical for $\beta>\beta_{\text{crt}}:=\min_x \beta_{\text{crt}}^{(x)}$. If every sector threshold is infinite, the Gibbs state is a stabilizer mixture at all temperatures.

Figure~\ref{F:spin-systems} shows this behavior for three benchmark Hamiltonians. The models test different aspects of the mechanics and are motivated by experimental methods for engineering high-order spin interactions:
\begin{align}
    H_1&=-J\sum_{i=1}^9Z_i Z_{i+1}-gX_1 X_2 \cdots X_{10} -gZ_1 ,\nonumber\\
    H_2 &=J\sum_{i=1}^{3}Z_iZ_{i+1}-g\qty(X_1X_2X_3X_4+Y_1X_2X_3X_4+Z_1),\\
    H_3 &=-J\sum_{i=1}^{4}Z_iZ_{i+1}-h\sum_{i=1}^{5}Z_i-g\qty(e^{i\phi}\prod_{i=1}^{5}\sigma_i^++e^{-i\phi}\prod_{i=1}^{5}\sigma_i^-), \nonumber
\end{align}
where $\sigma_i^\pm=\tfrac12(X_i\pm iY_i)$. The first model, $H_1$, is a ten-qubit repetition-code used to test scalability. The operators $Z_i Z_{i+1}$ are the parity checks of the code, while $\prod X_i$ and $Z_1$ act as logical $X$ and $Z$ in every sector. All $2^9$ sectors consequently experience the same two-component logical field, showing why increasing the Hilbert space dimension does not increase the RoM. Programmable $N$-body interactions in trapped-ion systems offer a way to implement its global coupling~\cite{Katz2023programmable}. The Hamiltonian $H_2$ adds a third logical-field component of equal strength and reaches the maximal RoM value of $\sqrt{3}$ in the zero-temperature limit. Four-body interactions of this type have been demonstrated in trapped-ion processors~\cite{Katz2023}. The third model $H_3$, behaves differently: its exchange term generates coherence only within the GHZ sector. Its contribution to the global RoM therefore depends on its thermal population. This model is motivated by a five-body exchange process implemented in a superconducting circuit~\cite{Zhang2022}.

\begin{figure}[t]
    \centering
    \includegraphics{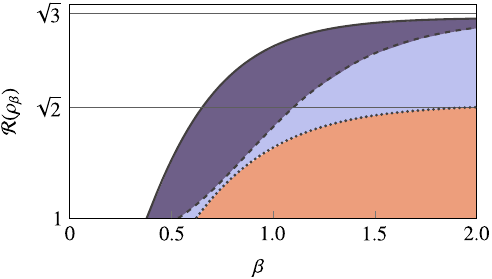}
    \caption{\emph{Magic at thermal equilibrium}. Robustness of magic of the Gibbs state as a function of the inverse temperature for three spin models. The dotted, solid and dashed curves correspond to $H_1$, $H_2$ and $H_3$. We set $J\!=\!g\!=\!1$, while for $H_3$ we choose $\phi=\tfrac\pi4$ and $h\!=\!\tfrac{g}{5\sqrt{2}}$. Each state has $\mathcal{R}(\rho_\beta)=1$ until its effective Bloch vector in at least one sector crosses the boundary of the stabilizer octahedron, at $\beta_{\mathrm{crt}}g\simeq0.380$, $0.538$, and $0.623$ for $H_2$, $H_3$, and $H_1$, respectively. As $\beta\rightarrow\infty$, the RoM of the Gibbs state of $H_1$ approaches $\sqrt{2}$, whereas those of $H_2$ and $H_3$ approach the universal bound $\sqrt{3}$.}
    \label{F:spin-systems}
\end{figure}

\paragraph*{Do all symmetries suppress magic?--} Symmetries have been shown to reduce the computational complexity of characterizing nonstabilizerness~\cite{Passarelli2024,Heinrich2019robustnessofmagic,varela2026predicting}, but they do not always suppress the resource itself. For example, permutation symmetry reduces the number of distinct Pauli expectation values to $\mathcal{O}(N^3)$ operations~\cite{Passarelli2024}, yet $\ket{T}^{\otimes N}$ still has a RoM that grows without bound with $N$, where $\ket{T}=2^{-\nicefrac12}(\ket{0}+e^{i\nicefrac\pi4}\ket{1})$. This naturally raises the question, \textit{how does symmetry suppress magic in $\ms X$-states}? As we have noted, $\ms X$-states are supported on $\bigoplus_{x \in \mathcal{X}_N} \mathcal{B}_x$. Using stabilizer operations such as controlled-not operations conditioned on a single qubit, any $\ms X$-state can be transformed into the form $\tilde{\rho}_{\ms X} = \sum_{s \in \{0,1\}^{N-1}} \ketbra{s}{s} \otimes \sigma_s$, where $\sigma_s$ are subnormalized single-qubit states, representing a classical mixture of states on a single \textit{logical} qubit (see~\hyperref[Sec:Endmatter]{End Matter}). This construction explains why the RoM of $\ms X$-states cannot exceed $\sqrt{3}$ and elucidates the two key ingredients required for suppression here. (i) A symmetry that block-diagonalizes the state, and (ii) a Clifford operation that decodes each block into a computational-basis label and a single qubit. Together, these two conditions limit the magic of the state to the maximum allowed in a single block. Interestingly, genuine multipartite entanglement remains possible in $\ms X$-states and is not similarly suppressed~\cite{huber_12}, as an effective qubit subspace such as $\operatorname{span}\Bigl\{\ket{0^N},\ket{1^N}\Bigl\}$ can still host a maximally entangled state $\ket{\text{GHZ}_N}$. Note that tracing out a single qubit from any $\rho_{\ms X}$ destroys all coherence, eliminating both the genuine multipartite entanglement and magic. 

\paragraph*{Discussion--} The nearest neighbor parity symmetry also has a coding interpretation. The subspace symmetries $G_i$ are the stabilizers of the $N$-qubit repetition code, the simplest stabilizer code, and define its stabilizer Hamiltonian $H_{\text{Rep.}} = - \sum^{N-1}_{i=1} G_i$~\cite{Kitaev2003,stab_ham}. The Hamiltonians studied here are perturbations of this stabilizer Hamiltonian that preserve its symmetry and can generate magic below a critical temperature. The bound on magic then follows from the fact that the repetition code contains only one logical qubit. This suggests asking whether perturbing the stabilizer Hamiltonians of other codes gives physically interesting models whose magic is also limited by symmetry. Exact formulas for the RoM are likely to be harder to find when the symmetry is associated with a stabilizer code with a larger logical subspace, since explicit expressions for the RoM of arbitrary multi-qubit states remain unknown. 

We asked whether a physically motivated symmetry can make many-body nonstabilizerness analytically solvable. Parity symmetry among neighboring spins does so. The same symmetry also limits the magic to the single-qubit maximum of $\sqrt{3}$ for the RoM~\cite{Howard2017}. Both results follow because parity symmetry makes the state block diagonal, with each block describing an \emph{effective} qubit whose nonstabilizerness is bounded. For many-body systems with this symmetry, the magic-dependent sampling overhead in quasiprobability simulations~\cite{Pashayan2015,delfosse_15} is therefore bounded by a constant independent of system size. Efficient simulation also requires efficient sampling of the stabilizer decomposition and efficient simulation of the remaining operations and measurements. Complementing this result, we turned to the dynamics and derived the conditions under which equilibrium states have this symmetry. We show that magic in these Gibbs states requires global transverse coupling acting on every site, and derive the corresponding critical temperature. This complements results on magic generation under thermal operations~\cite{deOliveiratrading} and its temperature dependence in specific models~\cite{sabharwal2026stuffmagiconcompact,zavatti2026quantummagicstronglycorrelated}, adding an exact critical temperature for our symmetry class to general high-temperature bounds~\cite{putterman2026quantumthermalstateslook}.

\paragraph*{Acknowledgments--} The authors thank Darlan A. Moreira for the data on the transverse-field Ising model ($N=4$), and the High-Performance Computing Center at UFRN for the computational resources used to generate Fig.~\ref{F:X-twirl-Ising}. Financial support was provided by the Simons Foundation (Grant No. 1023171, R.C.), the Brazilian National Council for Scientific and Technological Development (CNPq; Grants No. 403181/2024-0 and 301687/2025-0), the EU Horizon Europe programme (QSNP, Grant No. 101114043), the National Institute of Science and Technology for Applied Quantum Computing through CNPq Process No. 408884/2024-0, and the Financiadora de Estudos e Projetos (Grant No. 1699/24 IIF-FINEP). R.C. thanks the Technical University of Denmark for its hospitality, where part of this work was carried out during a guest professorship supported by the Otto M\o nsted Foundation. J.X. acknowledges support from the European Research Council (ERC Project ’Cocoquest’ 101043705). 

\textit{Tool Disclosure \& Code Availability.} The authors used Mathematica 14.3 and ChatGPT 5.6 Sol for symbolic analysis. Owing to the first author's stubbornness, AI was not used to help prove the main result. The code used to generate the plots is available at~\href{https://github.com/AdeOliveiraJunior/When-Symmetry-Suppresses-Magic}{GitHub}.

\bibliography{2-citations}
\section*{End Matter}\label{Sec:Endmatter}
\subsection{Decoding \ms X-states}

A general $\ms X$-state can be expressed as 
\begin{align*}
    \rho_{\ms X} = \sum_{x \in \mathcal{X}_N} r_x \ketbra{x}{x} + r_{\bx} \ketbra{\bx}{\bx} + \alpha_x\ketbra{x}{\bx} + \alpha^*_x \ketbra{\bx}{x}, 
\end{align*}
or alternatively, writing $\alpha_s :=\langle0\, s|\rho_{\ms{X}}|\bar s \, 1\rangle$,
\begin{align*}
    \rho_{\ms X} \!=\! \!\sum_{s \in \{0,1\}^{N-1}}\! r_{s0} \ketbra{s \,0}{s\,0} \!+\! r_{\bar s1} \ketbra{\bar s \, 1}{\bar s \, 1} \!+\! \alpha_s\ketbra{s \, 0}{\bar s \, 1} \!+\! \alpha^*_s \ketbra{\bar s \, 1}{s \, 0}.
\end{align*}
If, conditioned on the last qubit, we apply an $X$ to each of the first $N-1$ qubits, the resulting Clifford operation is $C_N X = \prod ^{N-1}_{i =1} C_{N}X_i$, where $C_{N}X_i$ is the controlled-not with control $N$ and target $i$. We obtain
\begin{align*}
    \!C_N X \rho_{\ms X} (C_N X)^\dagger \!=\! \sum_{s \in \{0,1\}^{N-1}} \ketbra{s}{s} \otimes \sigma_s \text{ where } \sigma_s = \begin{pmatrix} r_{s0} && \alpha_s \\ \alpha^*_s && r_{\bar s1},
    \end{pmatrix}
\end{align*}
Here, $\sigma_s$ is subnormalized. Only the $\bar{s}$ terms are flipped and we see that the contents of $\rho_{\ms X}$ are now encoded in the last qubit as a classical mixture. Note that each controlled-not is a stabilizer unitary. One can think of this as the outcome of a repetition code where $\ket{0},\ket{1}$ of the last qubit are mapped to $\ket{0^N},\ket{1^N}$ and the string $s$ then tracks the error syndrome. Tracing over the first $N-1$ qubits we obtain 
\begin{gather*}
    \sigma' \!=\! \tr_{1\ldots N-1} \left\{\sum_{s \in \{0,1\}^{N-1}} \ketbra{s}{s} \otimes \sigma_s\right\} \!=\! \sum_{s \in \{0,1\}^{\times N - 1}} \sigma_s \!=\!  \begin{pmatrix} P && A \\ A^* && 1-P
    \end{pmatrix},
\end{gather*}
where $\sum_s r_{s\,0}= P, \sum_s \alpha_s= A$. Since all $\ket{s}$ are computational-basis states and hence stabilizer states, Clifford invariance and convexity gives $\mathcal{R}(\rho_{\ms{X}})\leq \sum_{s:\tr \sigma_s>0} (\tr \sigma_s) \mathcal{R}(\sigma_s/\tr\sigma_s)$ and the RoM is upper bounded by $\sqrt{3}$. Tracing out the sector label need not to preserve the RoM, so $\sigma'$ need not contain all the original magic.

\subsection{Numerical evaluation of the \ms{X}-twirl bound}
For the ferromagnetic transverse-field Ising Hamiltonian in Fig.~\ref{F:X-twirl-Ising}, 
for $N=4$, we optimize exhaustively by evaluating all $11\,475$ inequivalent Clifford twirls at each sampled $g$ and keeping the largest bound. Since Clifford operations that define the same twirl need not be checked separately, this gives the global optimum over Clifford operations. For $N=10$ and $N=15$, we instead evaluate the twirl bound using one fixed Clifford operation for each system size. The operations are selected via numerical searches at $g=1$ and then kept fixed for all values of $g$ shown. The selected circuits contain $58$ and $70$ elementary $H$, $S$ and CNOT gates, respectively. After applying each circuit to the ground state, we evaluate the bound using the closed-form sector expression for the RoM, without reconstructing the many-body density matrix. 

\subsection{Spin models}
For the three open chains in Fig.~\ref{F:spin-systems}, we compute the RoM from the exact sector expressions, including all sectors in the Gibbs state.
\begin{enumerate}
    \item For $N=10$, the logical field is $\v h_x=(g,0,g)$ in every sector giving $\mathcal{R}(\rho_\beta)=\max\biggl\{1,\sqrt{2}\tanh(\sqrt2|g|\beta)\biggl\}$.
    \item For $N=4$, the additional $Y_1 X_2 X_3 X_4$ term gives $\v h_x=(g,g,g)$ and $\mathcal{R}(\rho_\beta) = \max\biggl\{1,\sqrt3\tanh(\sqrt3|g|\beta)\biggl\} $
\end{enumerate}
In both models, the Ising coupling changes only the sector energies. Since all sectors have the same normalized logical state, the total RoM is independent of $J$.

\begin{enumerate}[start=3]
    \item For $N=5$, the exchange term couples only $\ket{0^5}$ and $\ket{1^5}$. In this ordered basis, the GHZ sector has energy offset $-4J$ and logical field $\v h_{\mathrm{GHZ}}=(g\cos\phi,-g\sin\phi, 5h)$. Writing $\Omega= \|\v h_{\mathrm{GHZ}}\|_2=\sqrt{g^2+25h^2}$, its population is $p_{\mathrm{GHZ}}(\beta) =\frac{2e^{4\beta J}\cosh(\beta\Omega)}{Z_{\beta}}$ and the RoM for $\Omega>0$ is 
\begin{equation*}
    \mathcal R(\rho_{\beta}) =1+p_{\mathrm{GHZ}}(\beta) \left[\frac{\|\v h_{\mathrm{GHZ}}\|_1}{\Omega} \tanh(\beta\Omega)-1\right]_+
\end{equation*}
\end{enumerate}
For $\Omega=0$, $\mathcal{R}(\rho_\beta)=1$. The remaining sectors are diagonal and have RoM one. 
\clearpage
\appendix
\onecolumngrid
\section{Encoded Pauli operators and encoded stabilizer states}\label{S:SM-1}

Recall that each sector $\mathcal{B}_x$ is two-dimensional and is identified with an effective qubit by defining the logical basis $\ket{0_L}:=\ket{x}$ and $\ket{1_L}=\ket{\bx}$. The corresponding effective Pauli operators are given by Eq.~\eqref{Eq:logical-paulis}, which act only within $\mathcal{B}_x$ and $\sigma_j^{(x)}=\Pi_x \sigma_j^{(x)}\Pi_x$ for $j=X,Y, Z$. These effective operators are obtained by restricting suitable physical Pauli strings to the sector:
\begin{equation}
    \begin{split}
    \sigma_X^{(x)}&=\Pi_x X^{\otimes N}\Pi_x,\\ 
    \sigma_Y^{(x)}&=(-1)^{x_1}\Pi_x Y_1 X_2 \cdots X_N\Pi_x,\\
    \sigma_Z^{(x)}&=(-1)^{x_1}\Pi_x Z_1 \Pi_x.
\end{split}
\end{equation}

The six encoded stabilizer states of sector $x$ are
\begin{equation}\label{Eq:six-octahedron-states}
\begin{matrix}
    \ket{x}, &  & & \ket{\bx}, &  & & \ket{e_X^{\pm}} = \frac{1}{\sqrt{2}}(\ket{x}\pm\ket{\bx}), &  & & \ket{e^{\pm}_Y} =\frac{1}{\sqrt{2}}(\ket{x}\pm i\ket{\bx}).
\end{matrix}
\end{equation}
Their projectors satisfy $\ketbra{e_X^{\pm}}=\tfrac12(\Pi_x\pm \sigma_X^{(x)})$ and  $\ketbra{e_Y^{\pm}}=\tfrac12(\Pi_x\pm \sigma_Y^{(x)})$. For the state $\ms{X}$-state in Eq.~\eqref{Eq:X-states}, we write $\alpha_R^{(x)}:=\operatorname{Re}(\alpha_x)$ and $\alpha_I^{(x)}:=\operatorname{Im}(\alpha_x)$. The sector coordinates are
\begin{equation}\label{Eq:encoded-pauli-components}
\begin{matrix}
\tr(\rho \sigma^{(x)}_X)= 2\alpha_R^{(x)}, &  & & \tr(\rho \sigma^{(x)}_Y)= -2\alpha_I^{(x)},  &  &  &\tr(\rho \sigma^{(x)}_Z)= r_x-r_{\bx}, &  &  & \tr(\rho \Pi_x)= r_x+r_{\bx}.  \\
\end{matrix}
\end{equation}
For every $x\in \mathcal{X}_N$, it is convenient to introduce the following quantities
\begin{equation}\label{Eq:effective-paulis}
\begin{matrix}
s_x:=|\alpha_R^{(x)}|+|\alpha_I^{(x)}|, &  & & p_x:=r_x+r_{\bx},  &  &  & d_x:=r_x-r_{\bx}, &  &  & m_x:=\min(r_x,r_{\bx}).   
\end{matrix}
\end{equation}
Thus $p_x$ is the sector population, $d_x$ the population bias, and $s_x$ the $\ell_1$ size of the coherence. We also set
\begin{equation}
    L_x:=2s_x+|d_x|,
\end{equation}
the $\ell_1$ length of the (subnormalized) encoded Bloch vector $\v v^{(x)} = (2\alpha_R^{(x)},-2\alpha_I^{(x)},d_x)$. We call $L_x\leq p_x$ the \emph{octahedron condition} for sector $x$: it states that the normalized encoded Bloch vector, when $p_x>0$, is in the one-qubit stabilizer octahedron. Since $p_x-|d_x|=2m_x$, this condition is equivalent to $s_x\leq m_x$. When $p_x=0$, positivity forces the entire sector block to vanish, and both conditions hold automatically.

Each sector $\mathcal{B}_x$ carries an encoded qubit and consequently contains an entire Bloch sphere of pure encoded states. However, it is not yet clear which encoded pure states are also $N$-qubit stabilizer states. In particular, are the six encoded eigenstates of the three logical Pauli operators the only pure stabilizer states supported on $\mathcal{B}_x$? The following Lemma answers this question affirmatively.

\begin{lem}[stabilizer states within an antipodal sector\label{Lem:stabilizer-states-antipodal}] The pure $N$-qubit stabilizer states supported on $\mathcal{B}_x$ are the six encoded octahedron states in Eq.~\eqref{Eq:six-octahedron-states}. Consequently, the encoded Bloch vector of any pure stabilizer state supported on $\mathcal{B}_x$ is one of $\pm \v e_X, \pm \v e_Y$, or $\pm \v e_Z$.
\end{lem}
\begin{proof}
We prove the statements in two steps. First, we verify that each of the six states in Eq.~\eqref{Eq:six-octahedron-states} is a pure $N$-qubit stabilizer state. We then show that no other pure stabilizer state can have its support contained in $\mathcal{B}_x$.

For the first step consider the pure stabilizer states, $\ket{0^N} = \ket{0\cdots 0}$, $\ket{1^N} =\ket{1\cdots 1}$ and $\ket{\text{GHZ}_N} = 2^{-\nicefrac12}(\ket{0\cdots 0}+\ket{1\cdots 1})$. Then it is easy to show that we may convert these states to $\ket{x},\ket{\bar{x}}, \ket{e^{\pm}_X}, \ket{e^{\pm}_Y}$ for any $x$ using only single qubit Clifford unitaries showing that these states are pure stabilizer states also. Now define the Pauli operator $X^x :=\bigotimes_{i=1}^N X_i^{x_i}$ which applies an $X$ at each position where $x_i =1$. It acts without introducing phases and satisfies $X^{x} \ket{0^N}= \ket{x}$ and $X^{x}\ket{1^N} = \ket{\bx}$. Therefore, it maps $\ket{0^N}$ and $\ket{1^N}$ to $\ket{x}$ and $\ket{\bar{x}}$, respectively. For the four remaining states consider applying the single-qubit Clifford gates $Z_1$ and $S_1=\operatorname{diag}(1,i)$, and their products to $\ket{\text{GHZ}_N}$. In particular $X^{x}\ket{\text{GHZ}_N} = \ket{e_X^+}$, $Z_1X^{x}\ket{\text{GHZ}_N} = \ket{e_X^-}$ and $S_1X^{x}\ket{\text{GHZ}_N} = \ket{e^+_Y}$, $Z_1S_1X^{x}\ket{\text{GHZ}_N} = \ket{e^-_Y}$ giving the six states in Eq.~\eqref{Eq:six-octahedron-states}. Since Clifford unitaries preserve stabilizer states, all six states in Eq.~\eqref{Eq:six-octahedron-states} are pure stabilizer states.

For the second step, let $\ket{\psi} = a\ket{x}+b\ket{\bx}$ be a pure stabilizer state supported on $\mathcal{B}_x$. If $ab =0$, then up to an irrelevant global phase, $\ket{\psi}$ is either $\ket{x}$ or $\ket{\bx}$. We then assume that $ab\neq 0$. Since we have just shown it is possible to map between sectors via Clifford unitaries it will be convenient to map the sector $\mathcal{B}_x$ to the canonical GHZ-sector, $\text{span}\{\ket{0^N},\ket{1^N}\}$ for this argument. Applying the Clifford operator $X^x$ introduced above gives $\ket{\varphi}:=X^x\ket{\psi}=a\ket{0\cdots0}+b\ket{1\cdots 1}$. Since Clifford unitaries preserve stabilizer states, $\ket{\varphi}$ is again a pure stabilizer state. Recall that if $\tau=\ketbra{\varphi}$ is a pure stabilizer state then denoting the stabilizer group $S$ we have $\tau = 2^{-N} \sum_{g\in S} g$. Consequently, for every Hermitian Pauli string $P$,
\begin{equation}\label{Eq:stabilizer-property}
    \langle\varphi|P|\varphi\rangle \in \{0,+1,-1\}.
\end{equation}
This follows from Hilbert-Schmidt orthogonality of Pauli strings, which says that the expectation vanishes unless $P=\pm g$ for some $g\in S$, in which case it equals $\pm 1$. We now apply this constraint to three Pauli strings that probe the three encoded Bloch components, namely $Z_1$, $X^{\otimes N}$, and $Y_1 X_2 \cdots X_N$. On the two-dimensional subspace $\operatorname{span}\{\ket{0\cdots0}, \ket{1 \cdots 1}\}$, these operators play the roles of the three Pauli operators of an encoded qubit. Their expectation values on $\ket{\varphi}$ are
\begin{align}\label{Eq:three-calculations}
\begin{split}
   \langle \varphi|Z_1|\varphi\rangle &= |a|^2-|b|^2 \\
     \langle \varphi|X^{\otimes N}|\varphi\rangle & = 2\operatorname{Re}(a^*b), \\
     \langle \varphi|Y_1X_2\cdots X_N|\varphi\rangle & = 2\operatorname{Im}(a^* b). \\
\end{split}
\end{align}
By Eq.~\eqref{Eq:stabilizer-property}, each of the three quantities in Eq.~\eqref{Eq:three-calculations} belongs to $\{0,\pm1\}$. On the other hand, normalization gives
\begin{equation}
    [2\operatorname{Re}(a^*b)]^2+[2\operatorname{Im}(a^*b)]^2+(|a|^2-|b|^2)^2 = (|a|^2+|b|^2)^2 =1.
\end{equation}
Thus, the three expectation values form a unit vector whose components all lie in $\{0,\pm 1\}$. It follows that exactly one component equals $\pm 1$, while the other two vanish. The possibility $|a|^2-|b|^2=\pm1$ would force either $a=0$ or $b=0$, contradicting $ab\neq0$. Hence $|a|^2-|b|^2=0$, so that $|a|=|b|=2^{-\nicefrac12}$. Moreover, one of $2\operatorname{Re}(a^*b)$ and $2\operatorname{Im}(a^*b)$ equals $\pm 1$, while the other vanishes. As a result, $\tfrac{b}{a} =\tfrac{a^*b}{|a|^2} = 2a^*b\in \{\pm 1, \pm i\}$. Hence $\ket{\varphi}$ is one of the four GHZ phase variants, and conjugating back by $X^x$ shows that $\ket{\psi}$ is one of the four coherent states in Eq.~\eqref{Eq:six-octahedron-states}.

Together with the two computational-basis cases, this proves that the only pure stabilizer states supported on $\mathcal{B}_x$ are the six states in Eq.~\eqref{Eq:six-octahedron-states}. Finally, the three Pauli expectations above, conjugated back by $X^x$ are exactly the three encoded Bloch components $\tr(\ketbra{\psi}\sigma_j^{(x)})$ for $j= X,Y,Z$. Their possible values are therefore precisely $\pm \v e_X$, $\pm \v e_Y$ and $\pm \v e_Z$, the six vertices of the encoded octahedron.
\end{proof}

The previous Lemma classifies pure stabilizer states that are supported entirely within a single sector. However, a general $N$-qubit stabilizer state may have support across several sectors. To analyse an $\ms{X}$-state sector by sector, we need to understand what such a global stabilizer state looks like after projection onto one fixed sector. The following Lemma shows that the projection is either zero, or after normalization, one of the six encoded stabilizer states identified above. Thus, in each sector, an arbitrary pure stabilizer state contributes an octahedron vertex scaled by its population in that sector. This is what allows the global stabilizer constraints to be translated into independent bounds on the encoded Bloch vector of each sector.

\begin{lem}[Antipodal-sector stabilizer projection\label{Lem:sector-stabilizer-projection}] For any pure $N$-qubit stabilizer projector $\tau \in \STAB_N$, and any $x\in \mathcal{X}_N$, let $\pi^{(x)}_\tau :=\tr(\tau\Pi_x)$ denote the population of $\tau$ in the sector $\mathcal{B}_x$, and define the corresponding unnormalized encoded Bloch vector components by $v_j^{(x)}(\tau):=\tr(\tau \sigma_j^{(x)})$ for $j= X,Y,Z$. Then, the unnormalized encoded Bloch vector $\v v^{(x)}(\tau) = (v_X^{(x)}(\tau), v_Y^{(x)}(\tau), v_Z^{(x)}(\tau))$ satisfies
\begin{equation}\label{Eq:bloch-decomposition}
    |v_X^{(x)}(\tau)| +|v^{(x)}_Y(\tau)|+|v_Z^{(x)}(\tau)| = \pi_\tau^{(x)}.
\end{equation}
Whenever $\pi_\tau^{(x)}>0$, the normalized projection of $\tau$ onto $\mathcal{B}_x$ is one of the six encoded stabilizer states in Eq.~\eqref{Eq:six-octahedron-states}.
\end{lem}
\begin{proof}
First consider the case in which the projection onto $\mathcal{B}_x$ has zero population. If $\pi_\tau^{(x)} =0$, then the positive operator $\Pi_x \tau \Pi_x$ has trace zero and consequently vanishes $\Pi_x \tau \Pi_x =0$. Since each encoded Pauli operator is supported on $\mathcal{B}_x$, we have $\sigma_j^{(x)} = \Pi_x\sigma_j^{(x)} \Pi_x$. Hence, for $j= X,Y,Z$, it follows that $v_j^{(x)}(\tau) = \tr(\tau \sigma_j^{(x)}) = \tr(\Pi_x\tau \Pi_x \sigma_j^{(x)}) =0$. Thus the claimed equality Eq.~\eqref{Eq:bloch-decomposition} is immediate. 

Suppose now that $\pi_\tau^{(x)}>0$. Recall that the projector onto $\mathcal{B}_x$ can be written as $\Pi_x = \prod_{i=1}^{N-1} \Pi_x^{(i)}$, where $\Pi_x^{(i)}:=\tfrac12(\iden + \varepsilon_i^x Z_i Z_{i+1})$. Each factor $\Pi_x^{(i)}$ is the projector onto the eigenspace of the Pauli observable $Z_i Z_{i+1}$ with eigenvalue $\varepsilon_i^x$. Since these Pauli observables commute, projecting onto $\mathcal{B}_x$ is equivalent to measuring them sequentially and conditioning on the outcomes $\varepsilon_1^x, \dots, \varepsilon_{N-1}^x$. We must verify that every step in this conditioning procedure occurs with nonzero probability. For $k=1, \dots, N-1$, define the partial projector $Q_k:=\Pi_x^{(1)}\dots \Pi_x^{(k)}$. The range of $\Pi_x$ is contained in the range of $Q_k$ because $\Pi_x$ imposes all $N-1$ syndrome conditions, whereas $Q_k$ imposes only the first $k$. Equivalently, $\Pi_x \leq Q_k$ as projectors. Consequently $\tr(Q_k \tau)\geq\tr(\Pi_x \tau) = \pi_\tau^{(x)} >0$. Thus, every intermediate sequence of outcomes has nonzero probability, so the corresponding conditional state is well defined. 

Conditioning a pure stabilizer state on a nonzero-probability outcome of a Pauli measurement produces another pure stabilizer state. Applying this fact successively shows that the normalized projection $\tilde{\tau}_x :=\tfrac{\Pi_x\tau\Pi_x}{\pi_\tau^{(x)}}$ is a pure stabilizer state supported entirely in $\mathcal{B}_x$. By Lemma~\ref{Lem:stabilizer-states-antipodal}, $\tilde{\tau}_x$ must therefore be one of the six encoded stabilizer states. Its encoded Bloch vector is consequently one of the six octahedron vertices and hence  $\sum_{j= X,Y,Z} \big|\tr(\tilde{\tau}_x \sigma_j^{(x)})\big|=1$. Finally, since each $\sigma_j^{(x)}$ is supported on $\mathcal{B}_x$, we have $v_j^{(x)}(\tau) = \tr(\tau \sigma_j^{(x)}) = \tr(\Pi_x \tau \Pi_x \sigma_j^{(x)})= \pi_\tau^{(x)}\tr(\tilde{\tau}_x\sigma_j^{(x)})$. Taking absolute values and summing over $j= X,Y,Z$ gives 
$\sum_{j= X,Y,Z} |v_j^{(x)}(\tau)| =\pi_\tau^{(x)}\sum_{j= X,Y,Z}|\tr(\tilde{\tau}_x\sigma_j^{(x)})| = \pi_\tau^{(x)}$. This proves the claimed equality when $\pi_\tau^{(x)}>0$.
\end{proof}

\section{Characterizing and quantifying nonstabilizerness}\label{S:SM-2}

We now have all the ingredients needed to determine whether an arbitrary $N$-qubit \ms{X}-state is stabilizer or magic. The following theorem gives a necessary and sufficient condition in terms of the encoded Bloch vector of each antipodal sector.

\begin{thm}[stabilizerness criterion for $N$-qubit \ms{X}-states \label{Thm:necessary-and-sufficient-condition}] Let $\rho$ be an $N$-qubit \ms{X}-state. Then, 
\begin{equation}
    \rho \in \conv(\STAB_N) \Longleftrightarrow s_x \leq m_x \quad \text{for every sector $x \in \mathcal{X_N}$}.
\end{equation}
    Thus, $\rho$ is magic if and only if at least one sector violates its octahedron condition.
\end{thm}
\begin{proof}
We prove the two implications separately. For sufficiency, we explicitly decompose each sector block into the six stabilizer-states projectors identified in Lemma~\ref{Lem:stabilizer-states-antipodal}. For necessity, we use Lemma~\ref{Lem:sector-stabilizer-projection} to show that every convex combination of global stabilizer states must satisfy the octahedron condition in each sector.

Suppose that $s_x\leq m_x$ for every sector $x\in \mathcal{X}_N$. We show that $\rho$ can be written as a convex combination of pure $N$-qubit stabilizer projectors. We start by fixing a sector $\mathcal{B}_x$ and defining the nonnegative coefficients $u_{\pm}:=|\alpha_R^{(x)}|\pm\alpha_R^{(x)}$, $w_{\pm}:=|\alpha_I^{(x)}|\mp\alpha_I^{(x)}$, $c_x:=r_x-s_x$ and $c_{\bx}:=r_{\bx}- s_x$. The coefficients $u_{\pm}$ and $w_{\pm}$ are nonnegative by construction, while $c_x$ and $c_{\bx}$ are nonnegative because $s_x\leq m_x$. Using $\ketbra{e_X^{\pm}}=\tfrac12(\Pi_x\pm\sigma_X^{(x)})$ and $\ketbra{ e_Y^{\pm}}=\tfrac12(\Pi_x\pm\sigma_Y^{(x)})$, a direct computation gives
\begin{equation}\label{Eq:projector-expansion}
    \Pi_x\rho\Pi_x = c_x\ketbra{x}+c_{\bx}\ketbra{\bx}+u_+\ketbra{e_X^+}+u_-\ketbra{e_X^-}+w_+\ketbra{e_Y^+} + w_- \ketbra{e_Y^-}.
\end{equation}
The total coefficient in this decomposition is $c_x + c_{\bx}+u_++u_-+w_++w_-= p_x -2s_x+2s_x=p_x$. By Lemma~\ref{Lem:stabilizer-states-antipodal} every projector in Eq.~\eqref{Eq:projector-expansion} belongs to $\STAB_N$. Summing the decompositions over all sectors gives a nonnegative decomposition of $\rho$ into pure stabilizer projectors. Its total weight is $\sum_{x\in \mathcal{X}_N} p_x=1$. It is therefore a convex decomposition, and hence $\rho \in \conv(\STAB_N)$.

Conversely, suppose that $\rho \in \conv(\STAB_N)$. By definition, there exist coefficients $\lambda_k\geq0$ with $\sum_k\lambda_k =1$ and pure stabilizer projectors $\tau_k$, such that $\rho=\sum_k \lambda_k \tau_k$, with $\tau_k\in \STAB_N$. Fix $x\in\mathcal{X}_N$. By Eq.~\eqref{Eq:encoded-pauli-components}, the $\ell_1$ length of the encoded Bloch vector in this sector is $L_x=\sum_{j= X,Y,Z}\qty|\tr(\rho \sigma_j^{(x)})|=2|\alpha_R^{(x)}|+2|\alpha_I^{(x)}|+|d_x|$.  Now inserting the stabilizer decomposition of $\rho$ into this expression and using linearity of the trace,
\begin{align}
    L_x &=\sum_{j= X,Y,Z} \qty|\tr(\rho \sigma^{(x)}_j)| = \sum_{j= X,Y,Z}\qty|\sum_k\lambda_k\tr(\tau_k\sigma_j^{(x)})| = \sum_{j= X,Y,Z}\qty|\sum_k\lambda_kv^{(x)}_j(\tau_k)|.
\end{align}
Using the triangle inequality and the fact that $\lambda_k\geq0$, we get $L_x \leq \sum_k \lambda_k \sum_{j= X,Y,Z}|v^{(x)}_j(\tau_k)|$. Now Lemma~\ref{Lem:sector-stabilizer-projection} applies to each pure stabilizer projector $\tau_k$, so
\begin{equation}
    L_x \leq \sum_k \lambda_k \pi_{\tau_k}^{(x)} = \sum_k \lambda_k\tr(\tau_k \Pi_x) = \tr[\qty(\sum_k \lambda_k \tau_k)\Pi_x] = \tr(\rho \Pi_x) = p_x.
\end{equation}
Thus every stabilizer mixture of the form in Eq.~\eqref{Eq:X-states} satisfies $L_x \leq p_x$. Since $L_x = 2s_x+|d_x|$ and $p_x -|d_x| = 2m_x$, this is equivalent to $s_x \leq m_x$. Since the sector $x$ is arbitrary, this condition must hold in every sector. This proves Theorem~\ref{Thm:necessary-and-sufficient-condition}.
\end{proof}

Theorem~\ref{Thm:necessary-and-sufficient-condition} answers whether the state is inside or outside the stabilizer polytope. We now ask a second question. Once the state is outside, how far is it? To answer this question we use the robustness of magic. For a state $\omega$, define
\begin{equation}
    \mathcal{R}(\omega) = \min \qty{\sum_k |l_k| : \omega = \sum_k l_k \tau_k, \tau_k \in \operatorname{STAB}_N}\geq 1.
\end{equation}
This is the minimal $\ell_1$ norm of a signed stabilizer decomposition. Since $\tr(\omega)=1$, one has $\mathcal{R}(\omega)\geq 1$, with equality if and only if $\omega \in \operatorname{conv}(\operatorname{STAB}_N)$. 

Our goal is to compute $\mathcal{R}(\rho)$ exactly. To this end, we will prove matching lower and upper bounds. We first derive the lower bound by writing the RoM in its LP-dual form:
\begin{equation}\label{Eq:dual-form-norm}
    \mathcal{R}(\omega) = \max_{H=H^{\dagger}} \qty{\tr(H\omega):|\tr(H\tau)|\leq 1 \:\forall \tau \in \operatorname{STAB}_N}.
\end{equation}
Thus, any Hermitian operator $H$ whose expectation lies between $-1$ and $1$ on every pure stabilizer projector gives a certified lower bound on $\mathcal{R}(\omega)$. 

To construct the upper bound, we decompose each sector block independently into its six encoded stabilizer states. We first need the optimal signed-decomposition cost of a general subnormalized qubit over the one-qubit stabilizer octahedron. The following Lemma gives this cost in closed form.

\begin{lem}[Subnormalized one-qubit octahedron norm]\label{Lem:subnormalized} Let $T=\frac{1}{2}(p \iden + a_x X + a_y Y+ a_z Z)$, where $p\ge0$ and $\v{a} = (a_x, a_y, a_z)\in\mathbbm{R}^3$, be a Hermitian one-qubit operator and set $L:=\|\v a\|_1$. Then,
\begin{equation}
    \min \qty{\sum_\ell |x_{\ell}|: T = \sum_\ell x_\ell \omega_{\ell}, \:\omega_l \in \operatorname{STAB}_1} = \max \{p,L\}.
\end{equation}
\end{lem}
\begin{proof}
Let $V_1:=\{\pm \v{e}_x, \pm \v{e}_y, \pm \v{e}_z\}$ be the six one-qubit stabilizer Bloch vectors and set $\v \sigma:=(X,Y,Z)$. For each $\v v_l\in V_1$, we define the corresponding one-qubit stabilizer projector $\omega_{\ell}:=\tfrac12(\iden +\v{v}_{\ell} \cdot \v{\sigma})$ for $\ell = 1, \ldots, 6$. Consider any signed stabilizer decomposition $T=\sum_{\ell} x_\ell \omega_\ell$. If a decomposition contains repeated vertices, we simply combine their coefficients, so it is enough to sum over the six vertices once. Expanding the right-hand side gives $T=\tfrac12[(\sum_\ell x_\ell)\iden +(\sum_\ell x_\ell \v{v}_\ell). \v\sigma]$. On the other hand, $T=\tfrac12(p\iden + \v a \cdot \v \sigma)$. By uniqueness of the Pauli expansion, the identity coefficients and the Bloch vector must agree. Consequently $\sum_\ell x_\ell = p$ and $\sum_\ell x_\ell \v{v}_\ell = \v a$. Consequently, $\sum_\ell|x_\ell| \geq|\sum_\ell x_\ell| = p$. Also, since $\|\v{v}_\ell\|_1 = 1$ for every stabilizer vertex, $L=\|\v a\|_1 = \|\sum_\ell x_\ell \v v_\ell\|_1 \leq \sum_\ell|x_\ell| \|\v v_\ell\|_1 = \sum_\ell|x_\ell|$. Thus, every signed stabilizer decomposition has total absolute coefficient at least $\max\{p,L\}$.

It remains to show that this lower bound can be attained. We consider two cases. First, suppose $L\le p$. For each $j= x, y, z$ with $a_j \neq 0$, put coefficient $|a_j|$ on the vertex $\operatorname{sgn}(a_j) \v e_j$. These terms contribute the Bloch vector $\sum_{j = x, y, z} |a_j|\operatorname{sgn}(a_j) \v e_j = \sum_{j = x, y, z} a_j \v e_j = \v a$ and their total coefficient is $\sum_{j = x, y, z} |a_j| =L$. Since $L\leq p$, there remains a nonnegative trace coefficient $p-L \geq 0$. We assign it equally to any antipodal pair of stabilizer vertices, for example by placing $\tfrac{p-L}{2}$ on $\v e_X$ and $-\v e_X$. This adds total coefficient $p-L$ but contributes zero Bloch vector. Therefore we obtain a nonnegative stabilizer decomposition with total coefficient $L+(p-L) = p$. Since all coefficients are nonnegative, its total absolute coefficient is also $p$. In this case $p=\max\{p,L\}$.

Now suppose $L>p$. Since $L>0$, define $t_j:=\tfrac{L-p}{2L} |a_j|$. Then $0\leq t_j \leq |a_j|$ and $\sum_{j = x, y, z} t_j = \tfrac{L-p}{2L}\sum_{j = x, y, z}|a_j| = \tfrac{L-p}{2}.$ For each $j$ with $a_j\neq0$, put coefficient $|a_j|-t_j$ on the vertex $\operatorname{sgn}(a_j) \v e_j$ and coefficient $-t_j$ on the opposite vertex $-\operatorname{sgn}(a_j)\v e_j$. If $a_j=0$, then $t_j=0$, and we simply omit that pair of vertices. Let us check what this signed decomposition gives. Its contribution to the Bloch vector in the $j$ direction is $(|a_j|-t_j)\operatorname{sgn}(a_j)\v e_j+(-t_j)(-\operatorname{sgn}(a_j)\v e_j)$. The second term equals $t_j \operatorname{sgn}(a_j)\v e_j$, so the total contribution is $|a_j|\operatorname{sgn}(a_j)\v e_j = a_j \v e_j$. Summing over $j = x, y, z$, the total Bloch vector is therefore $\v a$. The total trace coefficient is $\sum_{j = x, y, z}[(|a_j|-t_j)-t_j] = \sum_{j = x, y, z}(|a_j|-2t_j) = L-2\sum_{j = x, y, z}t_j = L-(L-p)=p$. Finally, the total absolute coefficient is $\sum_{j = x, y, z} (||a_j|-t_j|+|-t_j|)$. Since $0\leq t_j\leq|a_j|$, this becomes $\sum_{j = x, y, z}(|a_j|-t_j+t_j) = \sum_{j = x, y, z}|a_j|=L$. In this case, $L=\max\{p,L\}$, so the lower bound is again attained. 
\end{proof}

We now have all the tools needed to state our second main result. 

\begin{thm}[RoM for $N$-qubit \ms{X}-states \label{Thm:RoM-X-states}] For any $N$-qubit \ms{X}-state
\begin{align}\label{Eq:RoM-1}
   \mathcal{R}(\rho) &= \sum_{x\in \mathcal{X}_N} \max\{p_x, L_x\} = 1+2\sum_{x\in \mathcal{X}_N}(s_x-m_x)_+
\end{align}
where $(\bullet)_+ := \max\{0,\bullet\}$.
\end{thm}
\begin{proof}
We prove the theorem by establishing matching lower and upper bounds.

For the lower bound, we use the dual form in Eq.~\eqref{Eq:dual-form-norm}. For each sector $x\in\mathcal{X}_N$ and each sign vector $\v\epsilon^x=(\epsilon^x_X, \epsilon^x_Y, \epsilon^x_Z) \in \{\pm1\}^3$, define the Hermitian operator $W_x(\v \epsilon^x) := \sum_{j= X,Y,Z} \epsilon_j^x \sigma_j^{(x)}$. The operator $W_x$ is a signed combination of the three encoded Pauli operators in sector $\mathcal{B}_x$. For a subset $A\subseteq \mathcal{X}_N$, define $H_A:=\sum_{x\notin A} \Pi_x + \sum_{x\in A} W_x(\v \epsilon^x)$. Here $A$ is the set of sectors in which we use a signed encoded-Pauli witness. In the remaining sectors, we use the sector projector. The reason for this choice is that $\tr(\Pi_x\rho)=p_x$, whereas, after choosing the signs appropriately, $\tr(W_x(\v\epsilon^x)\rho)=L_x$. Thus $H_A$ will eventually allow us to select the larger of $p_x$ and $L_x$ in every sector. We first show that $H_A$ is dual feasible. Let $\tau\in\STAB_N$ and set $q:=\sum_{x\notin A} \pi_\tau^{(x)}$. Then $\sum_{x\in A} \pi_\tau^{(x)}=1-q$. By the triangle inequality and Lemma~\ref{Lem:sector-stabilizer-projection},
\begin{equation}
    \qty| \sum_{x\in A} \sum_{j= X,Y,Z}\epsilon_j^x v_j^{(x)}(\tau)| \leq \sum_{x\in A} \sum_{j= X,Y,Z}|v_j^{(x)}(\tau)| = \sum_{x\in A} \pi_\tau^{(x)} = 1-q. 
\end{equation}
On the other hand, $\tr(H_A\tau) = q+\sum_{x\in A} \sum_{j= X,Y,Z}\epsilon_j^x v_j^{(x)}(\tau)$. Since the second term lies in the interval $[-(1-q), 1-q]$ it follows that $\tr(H_A\tau) \in [2q-1,1]\subseteq [-1,1]$. Thus $|\tr(H_A\tau)|\leq 1$ for every pure stabilizer projector $\tau$, and therefore $H_A$ is feasible in the dual problem.

We now choose $A$ and the signs so as to maximise the expectation of $H_A$ on $\rho$. Let $A:=\{x\in\mathcal{X}_N: L_x\ge p_x\}$. Thus, in the sectors belonging to $A$, we use the Bloch witness $W_x$, while in the remaining sectors we use the projector $\Pi_x$. For every $x\in A$, choose the signs so that $W_x$ is aligned with the encoded Bloch vector $\v v^{(x)}=(2\alpha_R^{(x)},-2\alpha_I^{(x)},d_x)$. Explicitly, choose $\epsilon_X^x \alpha_R^{(x)}=|\alpha_R^{(x)}|$, $-\epsilon_Y^x \alpha_I^{(x)}= |\alpha_I^{(x)}|$ and $\epsilon_Z^x d_x = |d_x|$ with arbitrary choices when one of the corresponding quantities vanishes. Then $\tr(W_x(\v \epsilon^x)\rho)=2\epsilon_X^x\alpha_R^{(x)}-2\epsilon_Y^x\alpha_I^{(x)}+\epsilon_Z^xd_x = 2|\alpha_R^{(x)}|+2|\alpha_I^{(x)}|+|d_x| = L_x$. For $x\notin A$, $\tr(\Pi_x \rho) = p_x$. Consequently,
\begin{equation}
    \tr(H_A\rho) = \sum_{x\notin A} p_x + \sum_{x \in A} L_x = \sum_{x\in \mathcal{X}_N}\max\{p_x,L_x\}.
\end{equation}
Since $H_A$ is dual feasible, Eq.~\eqref{Eq:dual-form-norm} gives
\begin{equation}\label{Eq:RoM-lower-bound}
    \mathcal{R}(\rho) \geq \sum_{x\in \mathcal{X}_N}\max\{p_x,L_x\}.
\end{equation}

For the upper bound, we now construct a signed stabilizer decomposition attaining the same value. For each sector $x$, the block of $\rho$ supported on $\mathcal{B}_x$ can be written as $\Pi_x \rho \Pi_x =\tfrac12(p_x\Pi_x + \v a_x \cdot \v \sigma^{(x)})$, where $\v a_x=(2\alpha_R^{(x)},-2\alpha_I^{(x)},d_x)$ and $\v \sigma^{(x)} = (\sigma_X^{(x)},\sigma_Y^{(x)}, \sigma_Z^{(x)})$. Its encoded Bloch-vector $\ell_1$-length is $\|\v a_x\|_1 = 2|\alpha_R^{(x)}|+2|\alpha_I^{(x)}| + |d_x| = L_x$. Under the identification $\ket{0} \leftrightarrow \ket{x}$ and $\ket{1}\leftrightarrow \ket{\bx}$, this block is precisely a subnormalized one-qubit operator with trace $p_x$ and Bloch-vector length $L_x$. Lemma~\ref{Lem:subnormalized} gives a signed decomposition of $\Pi_x\rho\Pi_x$ into the six encoded octahedron projectors with total absolute coefficient $\max\{p_x, L_x\}$. By Lemma~\ref{Lem:stabilizer-states-antipodal}, all six encoded projectors are genuine $N$-qubit stabilizer projectors. Performing this decomposition independently in every sector and combining the resulting terms gives a signed stabilizer decomposition of $\rho = \sum_{x\in \mathcal{X}_N} \Pi_x \rho \Pi_x$, with total absolute coefficient $\sum_{x\in\mathcal{X}_N} \max\{p_x,L_x\}$. Hence $\mathcal{R}(\rho) \leq \sum_{x\in\mathcal{X}_N} \max\{p_x, L_x\}$. Combining this inequality with Eq.~\eqref{Eq:RoM-lower-bound} gives
\begin{equation}
    \mathcal{R}(\rho) = \sum_{x\in\mathcal{X}_N}\max\{p_x,L_x\}.
\end{equation}
Finally $\max\{p_x,L_x\} = p_x+(L_x-p_x)_+$. Since $L_x-p_x=2s_x+|d_x|-p_x = 2(s_x-m_x)$, where we used $p_x-|d_x|=2m_x$, we obtain 
\begin{equation}
    \mathcal{R}(\rho) = \sum_{x\in\mathcal{X}_N} p_x + 2\sum_{x\in\mathcal{X}_N} (s_x-m_x)_+ = 1+2\sum_{x\in \mathcal{X}_N}(s_x-m_x)_+.
\end{equation}
This establishes both equalities in Eq.~\eqref{Eq:RoM-1} and completes the proof of Theorem~\ref{Thm:RoM-X-states}.
\end{proof}

Theorem~\ref{Thm:RoM-X-states} gives a closed-form expression for the RoM of an $N$-qubit \ms{X}-state as a sum of independent sector contributions. This raises the question of whether the magic of such a state can grow without bound with the number of qubits. Remarkably, it cannot. The following corollary shows that the magic within this family is uniformly bounded. More precisely, the RoM admits a dimension-independent upper bound, and the states saturating this bound can be characterised exactly.

\begin{cor}[Uniform bound~\label{cor:uniform-bound}] Every $N$-qubit $\ms{X}$-state satisfies
\begin{equation}
    \mathcal{R}(\rho) \leq \sqrt{3}.
\end{equation}
Equality holds if and only if every populated sector, after normalization, is a pure encoded qubit with Bloch vector $\tfrac1{\sqrt{3}}(\pm1,\pm1,\pm1)$. Importantly, the maximal magic $\mathcal{R}(\rho)=\sqrt{3}$, is independent of $N$ and can be attained by pure states for every $N$, and by mixed states for $N\geq 2$.
\end{cor}
\begin{proof}
For each sector $x\in\mathcal{X}_N$, we write the corresponding block as $\Pi_x\rho\Pi_x = \tfrac12(p_x\Pi_x+\v a_x. \v\sigma^{(x)})$, where $\v a_x\!=\!(2\alpha_R^{(x)},-2\alpha_I^{(x)}, d_x)$. The trace of this block is $p_x$, while the $\ell_1$ length of its encoded Bloch vector is $L_x =\|\v a_x\|_1$. The two eigenvalues of the block are $\tfrac12(p_x\pm \|\v a_x\|_2)$, and consequently, positivity requires $\|\v a_x\|_2\leq p_x$. On the other hand, the Cauchy-Schwarz inequality gives $L_x =\|\v a_x\|_1 \leq \sqrt{3}\|\v a_x\|_2 \leq \sqrt{3}p_x$. Since $p_x\leq \sqrt{3}p_x$, it follows that $\max\{p_x,L_x\} \leq \sqrt{3}p_x$ for every sector. Applying Theorem~\ref{Thm:RoM-X-states} and using $\sum_{x\in\mathcal{X}_N} p_x =1$, we obtain 
\begin{equation}
    \mathcal{R}(\rho) = \sum_{x\in \mathcal{X}_N}\max\{p_x,L_x\} \leq \sqrt{3}\sum_{x\in\mathcal{X}_N}p_x = \sqrt{3}.
\end{equation}
We now characterise the equality case. Suppose first that $\mathcal{R}(\rho) =\sqrt{3}$. Since every sector satisfies $\max\{p_x,L_x\}\leq\sqrt{3}p_x$, and all sector contributions are nonnegative, equality in the total sum requires $\max\{p_x,L_x\}=\sqrt{3}p_x$ for every populated sector $x$, namely every sector with $p_x>0$. Because $p_x<\sqrt{3}p_x$, this is possible only if $L_x=\sqrt{3}p_x$. Consequently, both inequalities $\|\v a_x\|_1\leq \sqrt{3}\|\v a_x\|_2 \leq \sqrt{3}p_x$ must be saturated. Equality in the first inequality holds if and only if the three components of $\v a_x$ have equal absolute value. Equality in the second holds if and only if $\|\v a_2\|_2 = p_x$. The latter condition means that the normalized block $\hat{\rho}_x:=\tfrac{\Pi_x \rho \Pi_x}{p_x}$ has a unit encoded Bloch vector and is therefore pure. Combining the two equality conditions gives $\tfrac{\v{a}_x}{p_x}=\tfrac{1}{\sqrt{3}}(\epsilon_{x,1},\epsilon_{x,2},\epsilon_{x,3})$, with $\epsilon_{x,j}\in\{\pm1\}$. Thus, every populated sector is a pure encoded qubit whose Bloch vector is one of the eight cube vertices $\tfrac{1}{\sqrt{3}}(\pm 1, \pm 1, \pm 1)$.  Conversely, suppose that every populated sector has this form. Then $\|\v a_x\|_1=\sqrt{3}p_x$, and hence $\max\{p_x, L_x\}=\sqrt{3}p_x$ in every populated sector. Therefore $\mathcal{R}(\rho) = \sqrt{3}\sum_{x\in\mathcal{X}_N} p_x =\sqrt{3}$. 

Finally, the bound can be attained by choosing a single populated sector whose normalized encoded Bloch vector is $\tfrac{1}{\sqrt{3}}(\pm 1, \pm 1, \pm 1)$. This gives a pure \ms{X}-state with RoM $\sqrt{3}$. When $N\geq 2$, one may instead populate two or more orthogonal sectors, assigning to each a pure encoded state with such a Bloch vector. The resulting state is mixed, while each sector contributes $\sqrt{3}p_x$, so its total RoM remains $\sqrt{3}$.
\end{proof}

The bound above determines the largest amount of magic that can be attained within the family of $N$-qubit \ms{X}-states. A different question is whether the exact formula in Theorem~\ref{Thm:RoM-X-states} can also be useful beyond this family. The following Corollary answers this question affirmatively. Averaging an arbitrary $N$-qubit state over the even-weight $Z$-type Pauli operators removes all matrix elements except those on the diagonal and antidiagonal, thus producing an \ms{X}-state. Since this averaging is a convex mixture of Clifford conjugations, it cannot increase the RoM. The exact RoM of the resulting \ms{X}-state provides an analytic lower bound on the RoM of the original state. 

\begin{cor}[$X$-twirl lower bound for arbitrary $N$-qubit states] Let $\sigma$ be an arbitrary $N$-qubit state. Define the $\ms{X}$-twirl by
\begin{equation}
    \Lambda_{\ms{X}}(\sigma) :=\frac{1}{2^{N-1}}\sum_{w\in\mathcal{E}_N} Z_w \sigma Z_w\quad \text{with} \quad Z_w:=\bigotimes_{i=1}^NZ_i^{w_i}.
\end{equation}
Here, $\mathcal{E}_N:=\{w\in\{0,1\}^N:|w|\:\text{is even}\}$ is the subgroup of even-weight bit strings. The $N-1$ independent generators $G_i$ generate the Pauli operators $Z_w$ with $w\in\mathcal{E}_N$, so this is the twirl over the parity symmetry. Then $\Lambda_{\ms{X}}(\sigma)$ is an $\ms{X}$-state whose diagonal and antidiagonal coincide with those of $\sigma$. For every $N$-qubit Clifford unitary $U$,
\begin{equation}
    \mathcal{R}(\sigma)\geq 1+ 2\sum_{x\in \mathcal{X}_N}\qty[s_x(\Lambda(U\sigma U^\dagger)) - m_x(\Lambda(U\sigma U^\dagger))]_+.
\end{equation}
\end{cor}
\begin{proof}
We first show that the map $\Lambda$ removes every matrix element except those on the diagonal and antidiagonal. Let $y,z\in\{0,1\}^N$ label two computational-basis states. For any $w\in\mathcal{E}_N$, one has $Z_w\ket{y} = (-1)^{w\cdot y} \ket{y}$, where $w\cdot y = \sum_{i=1}^N w_i y_i$ is evaluated modulo two. Consequently $\langle y|Z_w \sigma Z_w|z\rangle = (-1)^{w\cdot y +w\cdot z} \langle y|\sigma|z\rangle = (-1)^{w\cdot(y\oplus z)}\langle y|\sigma|z\rangle$. Averaging over all even-weight strings gives $c(v)=\langle y|\Lambda_{\ms{X}}(\sigma)|z\rangle = c(v)\langle|\sigma|z\rangle$, where $v:=y\oplus z$ and $c(v):=\tfrac{1}{2^{N-1}}\sum_{w\in\mathcal{E}_N} (-1)^{w\cdot v}$. If $w\cdot v=0$ for every $w\in \mathcal{E}_N$, then every term in the sum equals one and $c(v) =1$. Otherwise, there exists $w_0\in\mathcal{E}_N$ such that $w_0 \cdot v=1$. Since $\mathcal{E}_N$ is a subgroup, the map $w\mapsto w\oplus w_0$ is a permutation of $\mathcal{E}_N$. Pairing the terms associated with $w$ and $w\oplus w_0$ gives $(-1)^{w\cdot v}+(-1)^{(w\oplus w_0)\cdot v} = (-1)^{w\cdot v}[1+(-1)^{w_0\cdot v}]=0$. Therefore $c(v)=0$ in this case.

It remains to determine which strings are orthogonal to every even-weight string. The even-weight subgroup is $\mathcal{E}_N=\{w\in\{0,1\}^N:\v 1\cdot w=0\}$, where $\v 1 :=(1,\cdots ,1)$. It has dimension $N-1$, and its orthogonal complement has dimension one. Since $\v 1\cdot w = |w|=0$ modulo two for every $w\in\mathcal{E}_N$, this orthogonal complement is $\mathcal{E}_N^{\perp} = \operatorname{span}\{\v 1\}=\{\v 0, \v 1\}$. It follows that
\begin{equation}
    c(y\oplus z) = \begin{cases*}
1\:\: y\oplus z = \v 0 \:\:\text{or}\:\:\v 1 \\
0, \:\: \text{otherwise}.
\end{cases*}
\end{equation}
The condition $y\oplus z=\v 0$ is equivalent to $z=y$, while $y\oplus z=\v 1$ is equivalent to $z=\bar{y}$. Hence $\Lambda_{\ms{X}}(\sigma)$ retains precisely the diagonal and antidiagonal. In particular, $\Lambda_{\ms{X}}(\sigma)$ is an \ms{X}-state, and its surviving matrix elements coincide with those of $\sigma$.

We now prove the RoM bound. Fix an arbitrary $N$-qubit Clifford unitary $U$ and define $\rho_U:=\Lambda_{\ms{X}}(U\sigma U^\dagger)$. Every $Z_w$ is a Pauli operator and as a result a Clifford unitary. Using the convexity of the RoM, we obtain
\begin{equation}
    \mathcal{R}(\rho_U) = \mathcal{R}\,\qty(\frac{1}{2^{N-1}}\sum_{w\in\mathcal{E}_N}Z_w U\sigma U^{\dagger}Z_w) \leq \frac{1}{2^{N-1}}\sum_{w\in\mathcal{E}_N}\mathcal{R}(Z_w U\sigma U^{\dagger}Z_w).
\end{equation}
The RoM is invariant under Clifford conjugation, so every term in the last sum satisfies $\mathcal{R}(Z_w U\sigma U^{\dagger}Z_w) = \mathcal{R}(U\sigma U^\dagger)=\mathcal{R}(\sigma)$. Since $|\mathcal{E}_N|=2^{N-1}$, it follows that $\mathcal{R}(\rho_U) \leq \mathcal{R}(\sigma)$. Finally, $\rho_U$ is an \ms{X}-state, so Theorem~\ref{Thm:RoM-X-states} gives
\begin{equation}
    \mathcal{R}(\rho_U)  = 1+ 2\sum_{x\in\mathcal{X}_N}[s_x(\rho_U)-m_x(\rho_U)]_+.
\end{equation}
Substituting $\rho_U = \Lambda_{\ms{X}}(U\sigma U^\dagger)$ and combining the last two equations gives,
\begin{equation}
    \mathcal{R}(\sigma) \geq 1+ 2\sum_{x\in\mathcal{X}_N}[s_x(\rho_U)-m_x(\rho_U)]_+.
\end{equation}
Since $U$ is arbitrary, the bound holds for every $N$-qubit Clifford unitary.
\end{proof}

\section{Magic at thermal equilibrium}
\label{S:SM-4}

The $\ms{X}$-state family is characterised by a parity symmetry, namely commutation with all adjacent $Z$ parities $Z_i Z_{i+1}$ for $i=1,\cdots,N-1$. Equivalently, a thermal state $\rho_\beta \propto e^{-\beta H}$ is an $\ms{X}$ state at all $\beta\ge0$ if and only if $[H,Z_i Z_{i+1}]=0$ for $i=1,\cdots, N-1$. Indeed, if $H$ commutes with each parity check, then so does $e^{-\beta H}$. Conversely, for any finite $\beta>0$, the Gibbs state is strictly positive and $H=-\beta^{-1}[\log\rho_\beta+(\log Z_\beta)\iden]$, where $Z_\beta:=\tr(e^{-\beta H})$. Hence commutation of $\rho_\beta$ with the checks implies commutation of $H$ with them. At $\beta=0$, the Gibbs state is maximally mixed for every $H$, so symmetry at that temperature alone imposes no restriction on $H$. This symmetry condition has a simple Pauli-string interpretation. Consider $P=P_1\otimes\cdots\otimes P_N$, with $P_i \in \{\iden, X, Y, Z\}$ and define $t_i =0$ when $P_i \in \{\iden, Z\}$ and $t_i = 1$ when $P_i\in\{X,Y\}$. Since $X$ and $Y$ anticommute with $Z$, whereas $\iden$ and $Z$ commute with it, 
\begin{equation}
    (Z_i Z_{i+1})P = (-1)^{t_i + t_{i+1}}P(Z_i Z_{i+1}).
\end{equation}
Therefore, $P$ commutes with every adjacent check if and only if $t_i = t_{i+1}$ for every $i$. Since conjugation by each check maps every Pauli string to itself up to a sign, linear independence of the Pauli strings implies that this condition must hold term by term in $H$. Every Pauli term in an $N$-qubit $\ms{X}$ Hamiltonian is of exactly one of the following types:
\begin{enumerate}
    \item A diagonal string containing only $\iden$ and $Z$.
    \item A full-support transverse string containing $X$ or $Y$ on every site.
\end{enumerate}
Consequently, an $\ms{X}$-shaped $k$-local Hamiltonian with $k<N$ is diagonal in the computational basis. Its Gibbs state is therefore a convex mixture of computational-basis stabilizer states and has RoM equal to one at every temperature. The most general Hermitian $\ms{X}$ Hamiltonian can thus be expanded as
\begin{equation}
    H_{\ms{X}} = \sum_{w\in\{0,1\}^N} a_w Z_w + \sum_{u\in\{0,1\}^N} b_u Q_u,
\end{equation}
where 
\begin{align}
    Z_w:=\bigotimes_{i=1}^N Z_i^{w_i} \quad \text{and} \quad Q_u:=\bigotimes_{i=1}^N\begin{cases}
        X_i, & u_i = 0, \\
        Y_i, & u_i = 1.
    \end{cases}
\end{align}
and all coefficients are real. The statement is unchanged under local Clifford changes of axes. If $C=\bigotimes_{i} C_i$ is a local Clifford and $CHC^\dagger$ is $\ms{X}$-shaped, then $H$ commutes with the rotated Pauli checks $C^\dagger Z_i Z_{i+1}C$. Local Clifford conjugation preserves both Pauli weight and RoM. 

Suppose the restriction of the Hamiltonian to sector $x$ is $H_x = \epsilon_x \iden - \v h_x \cdot \v \sigma_L$, where $h_x=\|\v h_x\|_2$. Here $\iden$ and $\v\sigma_L$ act on the two-dimensional sector and correspond to $\Pi_x$ and $\v\sigma^{(x)}$ in the full Hilbert space. We take $\beta\geq0$. Since $(\v h_x \cdot \v \sigma_L)^2=h_x^2\iden$, we have for $h_x>0$ that $e^{-\beta H_x} = e^{-\beta \epsilon_X}[\iden \cosh (\beta h_x) + \hat{\v h}_x\cdot \v\sigma_L \sinh (\beta h_x)]$, where $\hat{\v h}_x = \tfrac{\v h_x}{h_x}$. The sector population in the Gibbs state is then
\begin{equation}
    p_x(\beta) = \tr(\Pi_x \rho_\beta) = \frac{2e^{-\beta \epsilon_x}\cosh(\beta h_x)}{Z_\beta}.
\end{equation}
This expression also holds when $h_x=0$. After normalising inside the sector, one obtains the encoded-qubit state $\hat\rho_x = \tfrac{1}{p_x}\Pi_x \rho_\beta \Pi_x = \tfrac{1}{2}[\iden + \tanh(\beta h_x)\hat{\v h}_x \cdot \v \sigma_L]$ for $h_x>0$. Thus, its encoded Bloch vector is $\v r_x(\beta) = \tanh(\beta h_x)\hat{\v h}_x$. When $h_x=0$, the normalized sector state is instead $\hat\rho_x=\tfrac12\iden$ and its encoded Bloch vector vanishes at every temperature.

The one-qubit stabilizer polytope is the octahedron $\| \v r\|_1 \leq 1$. For an encoded qubit, the same criterion applies. Therefore, for $h_x>0$, $\mathcal{R}(\hat\rho_x) = \max\{1, \tfrac{\| \v h_x\|_1}{\| \v h_x\|_2}\tanh(\beta\|\v h_x\|_2)\}$. The sector becomes magical if and only if $\tfrac{\| \v h_x\|_1}{\| \v h_x\|_2}\tanh(\beta\|\v h_x\|_2)>1$. If at least two components of $\v h_x$ are non-zero, then $\|\v h_x\|_1 > \|\v h_x\|_2$, and the critical inverse temperature is
\begin{equation}
    \beta_{\text{crt}}^{(x)} = \frac{1}{\|\v h_x\|_2}\operatorname{atanh}\qty(\frac{\|\v h_x\|_2}{\|\v h_x\|_1}).
\end{equation}
The sector is magical precisely when $\beta>\beta_{\text{crt}}^{(x)}$. If at most one component of $\v h_x$ is non-zero, including $h_x=0$, the sector remains a stabilizer mixture at every temperature and we set $\beta^{(x)}_\text{crt}:=+\infty$. For $h_x>0$, the ratio $\|\v h_x\|_1/\|\v h_x\|_2\in[1,\sqrt{3}]$ determines the dimensionless threshold $\beta_{\text{crt}}^{(x)}h_x$ and the limiting sector RoM as $\beta\to\infty$.

For an $\ms{X}$-state composed of several sectors, the global formula is 
\begin{equation}
 \mathcal{R}(\rho_\beta) = \sum_x p_x \mathcal{R}(\hat\rho_x)= 1+ \sum_{x:h_x>0} p_x\qty[\frac{\| \v h_x\|_1}{\| \v h_x\|_2}\tanh(\beta\|\v h_x\|_2)-1]_+   
\end{equation}
All sector populations are strictly positive at finite $\beta$. Since every term in the last sum is nonnegative, the Gibbs state is magical if and only if at least one sector is magical. Thus its critical inverse temperature is
\begin{equation*}
    \beta_\text{crt}=\min_{x\in\mathcal{X}_N} \beta_\text{crt}^{(x)},
\end{equation*}
where $\mathcal{R}(\rho_\beta)>1$ if and only if $\beta>\beta_\text{crt}$. Each normalized sector RoM is non-decreasing with $\beta$. This does not itself imply monotonicity of the global RoM because the sector populations also depend on $\beta$.
\end{document}